\documentclass[letterpaper, 10 pt, conference]{ieeeconf}  		
\IEEEoverridecommandlockouts                              			
\usepackage{amsmath}
\usepackage{amsfonts}
\usepackage{xcolor}
\usepackage{graphicx}

\newtheorem{remark}{Remark}
\newtheorem{theorem}{Theorem}

\title{\LARGE \bf
Reduced-order Smith predictor for state feedback control with guaranteed stability*
}

\author{Jesus-Pablo Toledo-Zucco$^{1}$, Frédéric Gouaisbaut$^{2}$, and Gaetan Chapput$^{3}$
\thanks{*This work was supported by LAAS-CNRS, INSA-Toulouse and UPS}
\thanks{$^{1}$Jesus-Pablo Toledo-Zucco is with INSA-Toulouse and LAAS-CNRS, Toulouse, France.
        {\tt\small jptoledozu@laas.fr}}%
\thanks{$^{2}$Frédéric Gouaisbaut is with UPS and LAAS-CNRS, Toulouse, France.
        {\tt\small }}%
 \thanks{$^{3}$Gaetan Chapput is with LAAS-CNRS, Toulouse, France.
        {\tt\small }}%
}

\begin{document}

\maketitle
\thispagestyle{empty}
\pagestyle{empty}

\begin{abstract}
This article deals with the implementation of the Smith Predictor for state feedback control in state space representation. The desired control law, obtained using partial differential equations and backstepping control, contains an integral term that has to be approximated for implementation. In this article, we propose a new way to implement this control law using a dynamic controller. The control law is composed of a state feedback term and a dynamic term that approaches the integral term that has to be estimated for implementation. Using a Lyapunov functional, we provide sufficient conditions, in terms of a linear matrix inequality, to guarantee that the closed-loop system is stable when the proposed control law is applied. We use three examples, taken from the literature, to show the benefits of the proposed approach.
\end{abstract}


\section{INTRODUCTION}

Controlling dynamical systems submitted to an input delay has attracted a lot of research during decades.
Firstly, many industrial systems are affected by delays (transport delays, communication delays) and not taking them into account during the design of the control law can lead to poor performances or even instability of the closed loop.
Furthermore, the introduction of such delays in the closed loop changes the nature of the system from a finite dimensional system to an infinite dimensional system. These systems are generally difficult to stabilize with classical tools of the literature.
Robust approaches (see for instance the monograph \cite{Fridman2014}, where the delays are modeled like a structured uncertainty, have been a first method to systematically control a input time delay system. But theses approaches have proven be very conservative especially for large delays.
The Smith Predictor (SP) control technique \cite{Smith1957JournalCloser} was a revolutionary approach to deal with dead time. By incorporating a compensation (or predictor) term, the closed-loop system behaves in the way it has been designed without delay but delayed in the same amount of time as the input delay. By removing the delay from the closed-loop dynamics,  many techniques like PID, pole placement or robust control \cite{MirkinTadmor2002} have been combined with the Smith predictor to control the system. 

Recent extensions have been proposed to deal with long time delays \cite{Hagglund1992JournalPredictive}, \cite{Briones2021ConferenceGeneralized} and with unstable plants \cite{Rojas2023ConferenceGeneralized}, \cite{Marquez2024JournalSimple}.
In state space representation, the analog SP has been developed using Finite Spectrum Assignment (FSA) \cite{Kwon1980JournalFeedback}, \cite{Manitius1979JournalFinite}, a reduction technique \cite{Artstein1982JournalLinear} and more recently using Partial Differential Equations (PDEs) and backstepping techniques \cite{Krstic2008JournalBackstepping}, \cite{Krstic2009BookDelay}.
This last approach is very interesting since, by modeling the delay as a transport equation, it allows not only to design of the control law but also to find the associated Lyapunov functional. 

Nevertheless, the control laws obtained either from the FSA or the PDE-backstepping approach are infinite dimensional and depend, via an integral, on the past values of the control input and it should be approximated, a technique known as late lumping. One of the main complication when implementing the control law obtained either from the FSA or the PDE-backstepping approach is that the integral terms have to be approximated using a finite dimensional systems.
Indeed, as noted in \cite{vanAssche1999ConferenceSome}, classical approximations techniques like Newton-Costes or Padé approach (see also \cite{engelborghs2001limitations} can destabilize the system. In a comparison with different approximation techniques, the conclusion is that even the most accurate methods can fail, whereas the least accurate ones can preserve closed-loop stability.
In \cite{Mirkin2004}, the authors points out that the instability phenomenon stems for one hand from a poor approximation of the control law at high frequencies and for the other hand from low robustness to high-frequency uncertainties. Hence, the authors propose to design a strictly proper approximation of the integral term leading to a vanishing perturbation in high frequency leading to a stabilizing controller.
In \cite{MondieMichiels2003}, the same idea is employed: the distributed parameter term of the FSA is replaced by a sum of point wise lumped delayed value of the input, calculated from a fixed step method and is combined with the design of a law pass filter to avoid the instability of the closed loop.
In \cite{AuriolMorrisDiMeglio2019}, the authors propose a general framework to study the design of approximated control laws issued from a backstepping approach (see for instance \cite{VazquezAuriolBribiescaArgomedoKrstic2026} for a complete survey). Using some specific approximation schemes combined with a Lyapunov function, the authors proposes some sufficient conditions for the stability of the closed loop.

In this article, we profit of the PDE-backstepping approach to approximate the integral term using the Finite Element Method (FEM). It has been shown that, using the FEM, one can preserve the energy structure for a large class of PDEs \cite{Toledo2024JournalStructure}. In particular, for the transport equation, the finite-dimensional approximation provided in  \cite{Toledo2024ConferenceScattering} preserves the scattering structure of the energy, meaning that for any order of approximation, the finite-dimensional approximation remains stable.  

In this article, we propose a methodology to implement the control law proposed in \cite{Krstic2008JournalBackstepping} (the analog to the SP), using a dynamic controller. The control law is then composed of a state feedback plus a dynamic feedback, in which the dynamic feedback is composed of a finite-dimensional approximation of the PDE \cite{Toledo2024ConferenceScattering}. 

The article is organized as follows: in Section \ref{Sec:ProblemFormulatuion}, we set the problem and provide an explanation of the main results of this article. In Section \ref{Sec:Numerical}, we recall the numerical approximation of the PDE. In Section \ref{Sec:Stability}, we provided the main result of this article, which corresponds to the stability guarantees of the closed loop. In Section \ref{Sec:Example}, we provide three examples taken from the literature. Finally, in Section \ref{Sec:Conclusions}, we provide some conclusions and future work.

\subsection*{Notation}
We denote by $0_{a\times b}$ a null matrix of size $a \times b$. When a matrix is written with subscript as follows $\left[\quad  \right]_{a\times b}$, the subscript $a \times b$ precises the dimension of the matrix. 

\section{PROBLEM FORMULATION}\label{Sec:ProblemFormulatuion}

\subsection{System in open loop}

We consider the following linear system 
\begin{equation}\label{Eq:SYS1}
\begin{cases}
\dot{X}(t) = AX(t) + BU(t-D), \\ 
U(t+\xi+D) = u_0(\xi), \, t \in [-D,0],\\
X(0)=X_0
\end{cases}
\end{equation}
subjected to an input time delay of a constant value $D$, $X(t) \in \mathbb{R}^n$ is the state, $U(t) \in \mathbb{R}$ is the input, $A \in \mathbb{R}^{n\times n}$ is the state matrix, $B \in \mathbb{R}^{n\times 1}$ is the input matrix.  $X_0$ and $u_0(\zeta), \forall \zeta \in [0,D]$ are the initial conditions of the ODE part and delayed control part respectively.
We assume that $(A,B)$ is controllable and the matrix $A$ may be unstable.

The input-delayed system \eqref{Eq:SYS1} can be modeled as a cascade interconnection between an ODE and a transport equation as follows:
\begin{equation}\label{Eq:SYS2}
\begin{cases}
\dot{X}(t) = A X(t) + B u(0,t), \\
\dfrac{\partial u}{\partial t} (\zeta,t) = \dfrac{\partial u}{\partial \zeta} (\zeta,t), \\
u(D,t) = U(t),\\
u(\zeta,0) = u_0(\zeta) ,\quad  \forall \zeta \in [0,D]\\
X(0)=X_0
\end{cases}
\end{equation}
in which $\zeta \in [0,D]$ is the spatial variable, $u(\zeta,t) \in \mathbb{R}$ is the infinite-dimensional variable. Notice that the PDE in \eqref{Eq:SYS2} corresponds to the transport equation in which a signal that enters at $\zeta =D$ exits at $\zeta = 0$, $D$ seconds later.


\subsection{Desired system in closed loop}

Following the backstepping approach popularized by \cite{Krstic2008JournalBackstepping}, we propose a target system, i.e. a desired behavior in closed-loop given by the following dynamical system:

\begin{equation}\label{Eq:DesiredCL2}
\begin{cases}
\dot{X}(t) = (A+BK) X(t) + B w(0,t), \\
\dfrac{\partial w}{\partial t} (\zeta,t) = \dfrac{\partial w}{\partial \zeta} (\zeta,t), \\
w(D,t) = 0,\\
w(\zeta,0) = w_0(\zeta) ,\quad  \zeta \in [0,D]\\
\end{cases}
\end{equation}
in which $K\in \mathbb{R}^{1 \times n}$ is a state feedback gain designed such that $A+BK$ is Hurwitz \textcolor{black}{and $w(\zeta,t)$ is the following backstepping transformation (See \cite[Section~4]{Krstic2008JournalBackstepping}):
\begin{equation*}
    w(\zeta,t) = u(\zeta,t)- \int_0^\zeta K e^{A(\zeta - y)}B u(y,t)dy - K e^{A\zeta}X(t).
\end{equation*}} 
We remark that for the undelayed case, the computation of $K$ can be achieved using pole-placement or linear quadratic regulator techniques.

In that case, the proposed control law achieving the closed loop \eqref{Eq:DesiredCL2} has been designed by \cite[Section~4]{Krstic2008JournalBackstepping} and can be viewed as an extension in the state space domain of the Smith predictor control:
\begin{equation}\label{Eq:ControlLaw}
\begin{split}
& U(t) = \int_0^D K e^{A(D-\zeta)}Bu(\zeta,t)d\zeta + K e^{AD} X(t), \\
& \dfrac{\partial u}{\partial t} (\zeta,t) = \dfrac{\partial u}{\partial \zeta} (\zeta,t), \quad u(\zeta,0) = u_0(\zeta) ,\quad  \zeta \in [0,D]\\
& u(D,t) = U(t).
\end{split}
\hspace{-0.4cm}
\end{equation}
with $u(\zeta,t)$ being an infinite-dimensional variable driven by the transport equation in a spatial domain $\zeta \in [0,D]$.
\subsection{Problem statement}

The paper proposes a new method to design an approximation of the infinite-dimensional control law \eqref{Eq:ControlLaw}.
The main idea, developed in Section 3, is first to construct an accurate and stable finite-dimensional approximation of the transport equation using the methodology proposed by \cite{Toledo2024ConferenceScattering}. This latter approximation is then combined with the initial control law \eqref{Eq:ControlLaw} to design a finite-dimensional dynamical control system stabilizing the delayed system.
In Section 4, the overall closed-loop system is proved to be asymptotically stable using a Lyapunov approach combined with the projection method introduced by \textcolor{black}{\cite{Seuret2015JournalHierarchy}}.

\subsection{Main results}
Inspired by recent results in stability of delay systems \cite{BajodekGouaisbautSeuret2024}, \cite{Seuret2015JournalHierarchy}  and by structure-preserving numerical approximation methods for PDEs \cite{Toledo2024ConferenceScattering,Toledo2024JournalStructure}, in this article, we proposed the following:
\begin{itemize}
\item An structure-preserving spatial discretization approach to approximate the control law in \eqref{Eq:ControlLaw}, obtaining a dynamic controller presented in the form of an ODE. 
\item A stability criterion, based on LMIs, to guarantee that the proposed control law stabilizes the original infinite-dimensional system \eqref{Eq:SYS1}.
\end{itemize}

\section{Numerical approximation of the control law}\label{Sec:Numerical}

In this section, we recall the structure-preserving discretization method proposed in \cite{Toledo2024ConferenceScattering} and we apply it to the desired control law in \eqref{Eq:ControlLaw}. 

\subsection{Numerical approximation of the transport equation}
 We follow the approach proposed in \cite{Toledo2024ConferenceScattering}. Firstly, a finite dimensional linear system \begin{equation}\label{Eq:ODE}
ODE\begin{cases}
E_d \dot{\hat{u}}_d(t) = A_d \hat{u}_d(t) + B_d U(t), \forall t> 0\\
\hat{u}_d(0) = \hat{u}_{d0},
\end{cases}
\end{equation} is designed to approximate the transport equation:
\begin{equation}\label{Eq:PDE}
PDE\begin{cases}
\dfrac{\partial u}{\partial t} (\zeta,t) =  \dfrac{\partial u}{\partial \zeta} (\zeta,t), \\ u(\zeta,0) = u_0(\zeta),\, \zeta \in [0,D]\\
u(D,t) = U(t), \\
\end{cases}
\end{equation}
where $\hat{u}_d(t) \in \mathbb{R}^N$ a finite-dimensional state variable with initial condition $\hat{u}_{d0}$ and constant matrices $E_d \in \mathbb{R}^{N\times N}$, $A_d \in \mathbb{R}^{N\times N}$, and $B_d \in \mathbb{R}^{N\times 1}$. The integer $N$ is a design parameter and is strongly related to the quality of the approximation. As soon as $N$ increases, the solution of the approximated ODE approaches the one of the PDE.

A first step, which is similar to the well known method of separation of variables,  is to introduce an approximation $\hat{u}(\zeta,t)$ of $u(\zeta,t)$ defined as 
\begin{equation}\label{Eq:Interpolation}
\hat{u}(\zeta,t) := \phi(\zeta)^\top \hat{u}_d(t),
\end{equation}
where $\phi(\zeta)= \begin{bmatrix}\phi_1(\zeta),&\cdots &,\phi_N(\zeta) \end{bmatrix}^\top$ forms a (finite) spatial basis on which the state of the transport equation is projected.\\
Then, we enforce the approximated variable $\hat{u}(\zeta,t)$ to satisfy the same equations of the PDE in \eqref{Eq:PDE}, that is:
\begin{equation}\label{Eq:PDEApp}
PDE\, (Approx.)\begin{cases}
\dfrac{\partial\hat{u}}{\partial t} (\zeta,t) =  \dfrac{\partial \hat{u}}{\partial \zeta} (\zeta,t), \\ \hat{u}(\zeta,0) = u_0(\zeta), \\
\hat{u}(D,t) = U(t). 
\end{cases}
\end{equation}
We replace \eqref{Eq:Interpolation} in \eqref{Eq:PDEApp}:
\begin{equation}\label{Eq:PDEApp2}
PDE\, (Approx.)\begin{cases}
\phi(\zeta)^\top\dfrac{\partial\hat{u}_d}{\partial t} (t) =  \dfrac{\partial \phi }{\partial \zeta} (\zeta) ^\top \hat{u}_d(t), \\ 
\phi(\zeta)^\top \hat{u}_d(0) = u_0(\zeta), \\
\phi(D)^\top \hat{u}_d(t) = U(t). 
\end{cases}
\end{equation}
Multiplying the two first equations in \eqref{Eq:PDEApp2} by $\phi(\zeta)$ at the left and integrating both equations over the spatial domain $\zeta \in [0,D]$, one obtains:
\begin{equation}\label{Eq:PDEApp3}
\left( \int_0^D \phi \phi ^\top d\zeta \right) \dot{\hat{u}}_d = \left( \int_0^D \phi  \dfrac{\partial \phi }{\partial \zeta} ^\top d\zeta \right) \hat{u}_d,
 \end{equation}
\begin{equation}\label{Eq:PDEApp4}
\left( \int_0^D\phi \phi ^\top d\zeta \right) \hat{u}_d(0) = \int_0^D \phi u_0 d\zeta, 
\end{equation}
where we have omitted the time $t$ and space $\zeta$ variables since it is clear from the context. Finally, integrating by part the right-hand side of \eqref{Eq:PDEApp3} and the use the third equation in \eqref{Eq:PDEApp2} we derive a finite dimensional system whose solutions are $\hat{u}_d$:
\begin{equation}
\begin{split}
\left( \int_0^D \phi \phi ^\top d\zeta \right) \dot{\hat{u}}_d =& \left( -\int_0^D \dfrac{\partial \phi }{\partial \zeta} \phi ^\top d\zeta - \phi(0)\phi(0)^\top \right) \hat{u}_d \\
&+ \phi(D) U(t),
\end{split}
 \end{equation}
 which is the ODE in \eqref{Eq:ODE} with matrices:
 \begin{equation}\label{Eq:Matrices}
 \begin{split}
 E_d &=  \int_0^D \phi \phi ^\top d\zeta, \\
 A_d &= -\int_0^D \dfrac{\partial \phi }{\partial \zeta} \phi ^\top d\zeta - \phi(0)\phi(0)^\top, \\
 B_d &= \phi(D), \\
 \end{split}
 \end{equation}
 associated with the initial conditions $\hat{u}_{d0} = E_d^{-1}\int_0^D\phi u_0 d\zeta$.
 \begin{remark}
 In \cite{Toledo2024ConferenceScattering} it has been shown that this numerical scheme preserves the energy structure of the initial PDE. This allows us to guarantee stability of the ODE in the Lyapunov (energy) sense using the approximated energy function of the PDE. In particular, this means that for any order of approximation $N$, the approximated ODE \eqref{Eq:ODE} is stable. 
 \end{remark}
 \subsection{Piecewise linear functions} \label{Sec:Piecewise}
 Several possibilities exit for the choice of the basis functions. Following \cite{Toledo2024JournalStructure}, we choose equidistant piece-wise linear functions as basis functions $\phi(\zeta)$. Considering a subdivision $h=\dfrac{D}{N-1}$ of the spatial variable $\zeta$, $\forall  \zeta \in [0,D]$, we define:
 \begin{equation}
 \begin{split}
     \phi_1(\zeta) &= \begin{cases}
     \tfrac{-1}{h}(\zeta - h),\quad \zeta \in [0,h] \\
     0, \quad \quad \quad \quad\quad elsewhere
     \end{cases} \\
          \phi_j(\zeta) &= \begin{cases}
     \tfrac{1}{h}(\zeta - (j-2)h),\quad \zeta \in [(j-2)h,(j-1)h] \\
     \tfrac{-1}{h}(\zeta - jh),\quad\quad\quad \zeta \in [(j-1),jh] \\
     0, \quad \quad \quad \quad\quad\quad\quad elsewhere
     \end{cases} \\
          \phi_N(\zeta) &= \begin{cases}
     0, \quad \quad \quad \quad\quad \quad\quad \quad elsewhere \\
     \tfrac{1}{h}(\zeta - (N-2)h),\quad \zeta \in [(N-2)h,(N-1)h] \\
     \end{cases} \\
 \end{split}
 \end{equation}
 Straightforward calculations give the following sparse matrices:
 \begin{equation}
     E_d = \dfrac{h}{6}\begin{bmatrix}
     2 &    1 &    0 &     \cdots &    0 \\
     1 &    4 &    1 &     \cdots &    0 \\
     0 &    1 &    4 &     \ddots &    0 \\
     \vdots &    \vdots &    \ddots &     \ddots &    \ddots \\
     0 &    0 &    0 &     \ddots &    2
     \end{bmatrix}
 \end{equation}
  \begin{equation}
     A_d = \dfrac{1}{2}\begin{bmatrix}
     -1 &    1 &    0 &     \cdots &    0 \\
     -1 &    0 &    1 &     \cdots &    0 \\
     0 &    -1 &    0 &     \ddots &    0 \\
     \vdots &    \vdots &    \ddots &     \ddots &    \ddots \\
     0 &    0 &    0 &     \ddots &    -1
     \end{bmatrix}, \quad B_d = \begin{bmatrix}
     0 \\ 0 \\ 0 \\ \vdots \\ 1
     \end{bmatrix},
 \end{equation}
with $E_d$, $A_d\in \mathbb{R}^{N\times N}$ and $B_d \in \mathbb{R}^{N\times 1}$.
\subsection{Dynamic control law}

Replacing $u(\zeta,t)$ by $\hat{u}(\zeta,t)$ from \eqref{Eq:Interpolation} in the equation of the backstepping controller \eqref{Eq:ControlLaw}, we obtain a new dynamical state feedback controller of the form:

\begin{equation}\label{Eq:DynamicController}
\Sigma_c\begin{cases}
\begin{split}
 \dot{\hat{u}}_d(t) &= \tilde{A}_{d}  \hat{u}_d(t) + \tilde{B}_d X(t),\\
 U(t) &= K_1 \hat{u}_d(t) + K_2 X(t).
\end{split}
\end{cases}
\end{equation}
with the following matrices:
\begin{equation}\label{Eq:K1K2}
\begin{split}
K_1 & = \int_0^D Ke^{A(D-\zeta)}B \phi^\top (\zeta) d\zeta, \\
K_2 & = K e^{AD}, \\
    \tilde{A}_d &= E_d^{-1}(A_d+B_dK_1), \\
    \tilde{B}_d &= E_d^{-1}B_dK_2.
\end{split}
\end{equation}

\section{Closed-loop stability in the Lyapunov sense}\label{Sec:Stability}
Even so the obtained ODE \eqref{Eq:ODE} is stable for any order of approximation $N$, the closed-loop stability of the system \eqref{Eq:SYS2} (or system \eqref{Eq:SYS1}) controlled by  \eqref{Eq:DynamicController} is not  guaranteed.
This section is dedicated to provide sufficient conditions to guarantee closed-loop stability of the proposed approach.

First of all, using the known matrices $A$ and $B$, one can design $K$ to assign the desired transient behaviour $A+BK$. Especially, one assume that $A+BK$ is Hurwitz.\\
Then, with the knowledge of the delay $D$, one can choose $1<N\in \mathbb{N}$ and define the matrices $E_d$, $A_d$ and $B_d$ as shown in Section \ref{Sec:Piecewise} and therefore the  dynamical state controller.

To improve the clarity of the main result, let us define the following matrices:
\begin{equation}\label{Eq:Kbar}
    \bar{K} = \begin{bmatrix}
K_2 & K_1 & 0_{1\times l}
\end{bmatrix}
\end{equation}
\begin{equation}\label{Eq:AA}
\mathcal{A} = \begin{bmatrix}  
A & 0_{n\times N} & 0_{n\times l} \\ 
\tilde{B}_d & \tilde{A}_d & 0_{N \times l} \\
0_{l\times n} & 0_{l\times N} & \tfrac{-1}{D}M
\end{bmatrix}, 
\end{equation}
\begin{equation} \label{Eq:BB}
\mathcal{B}_1 = \begin{bmatrix}  
0_{n\times 1}    \\ 
0_{N\times 1}  \\
L(D)   
\end{bmatrix},\quad \mathcal{B}_2 = \begin{bmatrix}  
 B   \\ 
 0_{N\times 1} \\
-L(0)
\end{bmatrix},\,  \mathcal{B} = \begin{bmatrix}
\mathcal{B}_1 & \mathcal{B}_2
\end{bmatrix},
\end{equation}
\begin{equation}\label{Eq:M}
M = \left[\begin{smallmatrix}
0 & 0 & 0 & 0 & 0 & 0 & 0 & \ddots  \\
2 & 0 & 0 & 0 & 0 & 0 & 0 &  \ddots   \\
0 & 6 & 0 & 0 & 0 & 0 & 0 &  \ddots   \\
2 & 0 & 10 & 0 & 0 & 0 & 0 &  \ddots   \\
0 & 6 & 0 & 14 & 0 & 0 & 0 &  \ddots  \\
2 & 0 & 10 & 0 & 18 & 0 & 0 &  \ddots   \\
0 & 6 & 0 & 14 & 0 & 22 & 0 &  \ddots   \\
\ddots & \ddots & \ddots & \ddots & \ddots & \ddots &  \ddots & \ddots  \\
\end{smallmatrix}\right]_{l \times l},
\end{equation}

\begin{equation}
L(D) = \left[\begin{array}{c}
1  \\
1    \\
1   \\
1  \\

\vdots  \\
\end{array}\right]_{l\times 1}, \quad L(0) = \left[\begin{array}{r}
1  \\
-1    \\
1   \\
-1  \\
\vdots  \\
\end{array}\right]_{l\times 1},
\end{equation}
\begin{equation}\label{Eq:Q}
Q = \begin{bmatrix}
1 & 0 & 0 & \cdots & 0 \\
0 & 3 & 0 & \cdots & 0 \\
0 & 0 & 5 & \cdots & 0 \\
\vdots & \vdots & \vdots & \ddots & \vdots \\
0 & 0 & 0 & \cdots & 2l-1 \\
\end{bmatrix}_{l\times l},
\end{equation}
\begin{equation}
\bar{Q} = \left[\begin{matrix}
0_{n \times n} & 0_{n \times N} & 0_{n \times l} \\
0_{N \times n} & 0_{N \times N} & 0_{N \times l} \\
0_{l \times n} & 0_{l \times N} & Q \end{matrix}\right],
\end{equation}
\begin{equation}\label{Eq:Psi}
\begin{split}
    \Psi =& \mathcal{A}^\top P + P \mathcal{A} + \alpha (1+D) \bar{K}^\top \bar{K}\\
    &  \quad \quad \quad - \tfrac{\alpha}{D} \bar{Q} + P \mathcal{B}_1 \bar{K}  + \bar{K}^\top \mathcal{B}_1^\top P ,
\end{split}
\end{equation}
\begin{equation}\label{Eq:Lambda}
\Lambda = \begin{bmatrix}
\Psi & P \mathcal{B}_2 \\
\mathcal{B}_2^\top P & -\alpha \end{bmatrix} .
\end{equation}
At this stage, we propose the following theorem:
\begin{theorem}
Consider the open-loop system \eqref{Eq:SYS2} (equivalently \eqref{Eq:SYS1}). If there exist an integer $l \in \mathbb{N}$, a symmetric matrix $P \in \mathbb{R}^{(n+N+l) \times (n+N+l)}$, and a scalar $\alpha$ such that the following LMI
\begin{equation}
    \begin{split}
        P &> 0 \\
        \alpha &> 0 \\
        \Lambda &< 0
    \end{split}
\end{equation}
is satisfy, then the closed-loop system between \eqref{Eq:SYS2} and \eqref{Eq:DynamicController} is asymptotically stable.
\end{theorem}
\begin{proof}
The closed-loop system between \eqref{Eq:SYS2} and \eqref{Eq:DynamicController} is given by the following ODE-PDE cascade:
\begin{equation}\label{Eq:Closed-Loop}
\begin{cases}
\dot{X}(t) = A X(t) + B u(0,t), \\
\dfrac{\partial u}{\partial t} (\zeta,t) = \dfrac{\partial u}{\partial \zeta} (\zeta,t), \quad \zeta \in [0,D]\\
u(D,t) = K_1 \hat{u}_d(t) + K_2 X(t), \\
\dot{\hat{u}}_d(t) = \tilde{A}_d \hat{u}_d(t) + \tilde{B}_d X(t).
\end{cases}
\end{equation}
To relate the infinite-dimensional state $u(\zeta,t)$ to the finite-dimensional ones $X(t)$ and $\hat{u}_d(t)$, we define the following vector
\begin{equation}
\Omega(t) = \begin{bmatrix}
\Omega_0(t)\\
\Omega_1(t) \\
\vdots \\
\Omega_{l-1}(t)
\end{bmatrix}, \quad \Omega_k(t) = \int_0^D L_k(\zeta)u(\zeta,t)d\zeta,
\end{equation}
with $L_k(\zeta)$ being the Legendre polynomials in $\zeta \in [0,D]$ defined as follows:
\begin{equation}
\begin{split}
    L_k(\zeta) &= (-1)^k \sum_{i = 0}^k p_i^k  \dfrac{\zeta^i}{D^i} , \\
    p_i^l &= (-1)^i \left(\begin{matrix}
k \\ i \end{matrix}\right)\left(\begin{matrix}
k +i\\ i \end{matrix}\right).
\end{split}
\end{equation}
We compute the dynamic of $\Omega(t)$ as follows (for simplicity, we omit time and spatial dependence):
\begin{equation}
\begin{split}
\dot{\Omega}_k &= \int_0^D \left[ L_k \dot{u} \right] d\zeta, \\
& = \int_0^D \left[L_k \dfrac{\partial u}{\partial \zeta } \right]d\zeta, \\
& = \int_0^D \left[\dfrac{\partial}{\partial \zeta}\left(L_k u \right) - \dfrac{dL_k}{d\zeta}u \right]d\zeta, \\
& = \left[L_k u \right]_{\zeta = 0}^{\zeta = D} -\int_0^D  \left[ \dfrac{dL_k}{d\zeta}u \right]d\zeta. \\
\end{split}
\end{equation}
Legendre polynomials satisfy:
\begin{equation}\label{Eq:dLdzProperty}
\dfrac{d L_k}{d \zeta} (\zeta) = \sum_{i = 0}^{k-1} \dfrac{2i+1}{D}\left( 1-(-1)^{k+i}\right)L_i(\zeta),
\end{equation}
which means that the derivative of a Legendre polynomial is defined as a linear combination of its previous Legendre polynomials (We refer to \cite[Section~3]{Seuret2015JournalHierarchy} for further information about Legendre polynomials and its properties). Then, by defining
\begin{equation}
L(\zeta) := \begin{bmatrix}
L_0(\zeta) \\
L_1(\zeta) \\
\vdots \\
L_{l-1}(\zeta) 
\end{bmatrix}
\end{equation}
we obtain
\begin{equation}
\dot{\Omega}(t) = L(D)u(D,t) -L(0)u(0,t) -\dfrac{1}{D}M \Omega (t), \\
\end{equation}
in which $M \in \mathbb{R}^{l \times l}$ is defined in \eqref{Eq:M} and is obtained using the property \eqref{Eq:dLdzProperty}.

We consider the following Lyapunov functional
\begin{equation}\label{Eq:Lyap}
V(t) = \eta (t)^\top P \eta (t) + \alpha \int_0^D (1+\zeta)u(\zeta,t)^2 d\zeta
\end{equation}
with $0<P \in \mathbb{R}^{(n+N+l) \times (n+N+l)}$, $\alpha >0$ and
\begin{equation}
\eta(t) = \begin{bmatrix}
X(t) \\
\hat{u}_d(t) \\
\Omega(t)
\end{bmatrix}.
\end{equation}
Notice that the projections $\Omega(t)$ are useful to relate the state of the PDE to the state of the ODEs through the matrix $P$. The dynamic of $\eta(t)$ is then, given by
\begin{equation}
\dot{\eta}(t) = \mathcal{A} \eta(t) + \mathcal{B}\begin{bmatrix} u(D,t) \\ u(0,t) \end{bmatrix}
\end{equation}
with $\mathcal{A}$ and $\mathcal{B}$ defined in \eqref{Eq:AA} and \eqref{Eq:BB}.
Now, we compute the time derivative of the Lyapunov functional \eqref{Eq:Lyap}. For simplicity, we omit time and spatial dependency and we denote $u_D = u(D,t)$ and $u_0=u(0,t)$. The computation is obtained as follows:
\begin{equation*}
\begin{split}
\dot{V}(t) &= \dot{\eta} ^\top P \eta + \eta^\top P \dot{\eta} + \alpha(1+D)u_D^2 - \alpha u_0^2 - \alpha\int_0^Du ^2d \zeta \\
&= {\eta} ^\top\left( \mathcal{A}^\top P + P \mathcal{A} \right) \eta + 2 \eta^\top P \mathcal{B} \begin{bmatrix} u_D \\ u_0
\end{bmatrix} \\
& \quad \quad \quad \quad + \alpha(1+D)u_D^2 - \alpha u_0^2 - \alpha \int_0^Du ^2d \zeta \\
&\leq {\eta} ^\top\left( \mathcal{A}^\top P + P \mathcal{A} \right) \eta + 2 \eta^\top P \mathcal{B} \begin{bmatrix} u_D \\ u_0
\end{bmatrix} \\
& \quad \quad \quad \quad + \alpha(1+D)u_D^2 - \alpha u_0^2 - \tfrac{\alpha}{D} \Omega^\top Q \Omega \\
\end{split}
\end{equation*}
where we have used the orthogonality property of Legendre polynomials (see \cite[Section~3]{Seuret2015JournalHierarchy} or Appendix \ref{Sec:Appendix}). Notice that $u(D,t) = u_D = \bar{K}\eta$, with $\bar{K}$ defined in \eqref{Eq:Kbar}. Then, the previous inequality becomes 
\begin{equation*}
\begin{split}
\dot{V}(t) &\leq  {\eta} ^\top\left( \mathcal{A}^\top P + P \mathcal{A}  \right) \eta+ 2 \eta^\top P \mathcal{B}_2 u_0 - \alpha u_0 ^2
 \\
& {\eta} ^\top\left( \alpha (1+D) \bar{K}^\top \bar{K} - \tfrac{\alpha}{D} \bar{Q} + P \mathcal{B}_1 \bar{K}  + \bar{K}^\top \mathcal{B}_1^\top P \right) \eta \\
\end{split}
\end{equation*}
or equivalently
\begin{equation}
\dot{V} \leq \begin{bmatrix}
 \eta \\ u_0 \end{bmatrix}^\top \begin{bmatrix}
\Psi & P \mathcal{B}_2 \\
\mathcal{B}_2^\top P & -\alpha \end{bmatrix} \begin{bmatrix}
 \eta \\ u_0 \end{bmatrix} = \begin{bmatrix}
 \eta \\ u_0 \end{bmatrix}^\top \Lambda \begin{bmatrix}
 \eta \\ u_0 \end{bmatrix}.
\end{equation}
with $\Psi$ and $\Lambda$ defined in \eqref{Eq:Psi} and \eqref{Eq:Lambda}, respectively. Since $\Lambda<0$ from the statement, then $\dot{V} \leq 0$, which concludes the proof.
\end{proof}


%
%

\section{Examples} \label{Sec:Example}
In this section, we study three examples from the literature. We use YALMIP and SDPT3 interfaces to solve the LMIs in Matlab. We also provide numerical simulations using the midpoint rule for time integration. In all simulations, the initial conditions $X(0)$, $\hat{u}_d(0)$ and $u(\zeta,0)$ are set to zero except for the first example in which we show the response to initial conditions. We add a reference signal $r(t)$ to the control law as follows:
\begin{equation*}
\Sigma_c\begin{cases}
\begin{split}
U(t) &= K_1 \hat{u}_d(t) + K_2 X(t) + Hr(t), \\
 \dot{\hat{u}}_d(t) &= \tilde{A}_{d}  \hat{u}_d(t) + \tilde{B}_d X(t) + E_d^{-1}B_d H r(t).
\end{split}
\end{cases}
\end{equation*}
with $H = -(C(A+BK)^{-1}B)^{-1}$ and $C \in \mathbb{R}^{1 \times n}$ defined for the output $y(t)=CX(t)$ to be tracked by the reference signal $r(t)$.

\subsection{Unstable first-order system with input delay \cite{vanAssche1999ConferenceSome}}\label{Example1}
In \cite{vanAssche1999ConferenceSome} the system \eqref{Eq:SYS1} with $X(t)\in\mathbb{R}$, $A = B=C =1$ with a delay $D =1$ has been used to show that the numerical integration of the integral term can destabilize the closed-loop system. The objective is to place the eigenvalue at $\lambda = -1$. To this end, $K = -2$ guarantees $A+BK = -1$. With the approach proposed in this article, the dynamic controller \eqref{Eq:DynamicController} of order $N =2$ and matrices:
\begin{equation}
    \tilde{A}_d = \begin{bmatrix}
        3.0000  &  5.8731 \\
   -9.0000      & -8.7463
    \end{bmatrix},
    \tilde{B}_d = \begin{bmatrix}
       10.8731 \\
  -21.7463
    \end{bmatrix}
\end{equation}
\begin{equation}
    K_1 = \begin{bmatrix}
-2.0000  & -1.4366
    \end{bmatrix},
    K_2 = 
-5.4366, H = 1,
\end{equation}
guarantees closed-loop stability with $l=4$. In Fig. \ref{fig:Example1}, we show the desired response and the obtained response with approximation of order $N=2$, $N=3$, and $N=10$. Between $t = 0$ and $t = 10$, the reference signal is set zero, and the closed-loop system evolution depends only on its initial conditions. Therefore,  as the delay $D=1$, the system \eqref{Eq:SYS2} is on open loop as the control $u(t)$ is equal to its initial conditions. The trajectories are still diverging from $t=0$ to $t=1$. Then the control law is applied and the states are converging to zero. At $t=10$, a step input is also applied and the trajectories are closed to the ideal trajectories of system \eqref{Eq:DesiredCL2}.
For every case, the closed-loop stability is guaranteed with the corresponding $l$ shown in Table \ref{Table:l}. Notice that for $N=10$ (as discussed in \cite{vanAssche1999ConferenceSome}), the response is almost superposed to the desired one.

\begin{figure}
\begin{center}
\includegraphics[width=8.4cm]{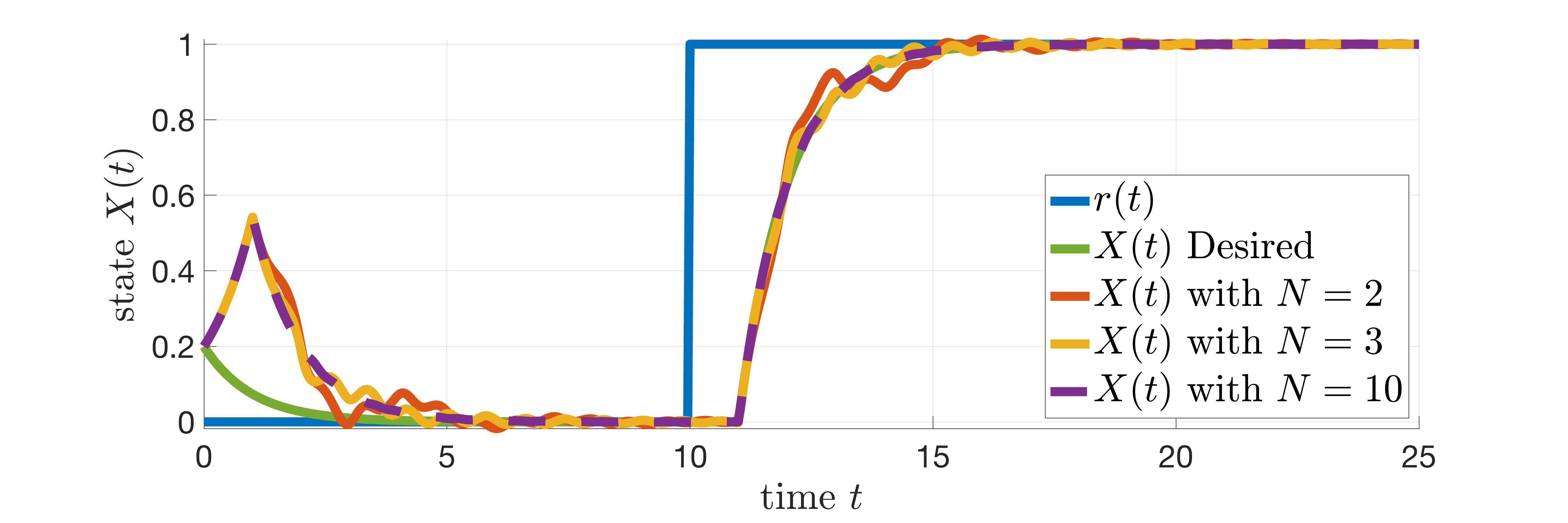}    
\caption{Closed-loop responses Example 1.} 
\label{fig:Example1}
\end{center}
\end{figure}

\subsection{Example 2. Unstable 3rd-order sys. with input delay \cite{Krstic2008JournalBackstepping}}\label{Example2}
We use the same example than in \cite[Section 4]{Krstic2008JournalBackstepping} the system \eqref{Eq:SYS1} with $X(t)\in\mathbb{R}^3$,
\begin{equation*}
A = \begin{bmatrix} 
2 &0 & 1\\
1& -2 & -2 \\
0 &1 &-1\end{bmatrix}, \quad 
B = \begin{bmatrix} 
0 \\ 0 \\1\end{bmatrix},\quad  C = \begin{bmatrix} 
1 & 0 & 0\end{bmatrix},
\end{equation*}
and $K$ designed using LQR with gains $Q = I_3$ and $R =1$. 
With a delay $D = 0.5$ and the case $N=2$ we obtain the following matrices for the dynamic controller \eqref{Eq:DynamicController}:
\begin{equation}
    \tilde{A}_d = \begin{bmatrix}
        7.6103 &  12.7876 \\
  -21.2206 & -19.5753
    \end{bmatrix},
    \end{equation}
    \begin{equation}
    \tilde{B}_d = \begin{bmatrix}
       188.0448  & 12.7732 &  52.8685 \\
 -376.0896  & -25.5463  & -105.7369
    \end{bmatrix}
\end{equation}
\begin{equation}
    K_1 = \begin{bmatrix}
-2.4026  & -1.6969
    \end{bmatrix}, H = 5.6125,
    \end{equation}
    \begin{equation}
    K_2 =  \begin{bmatrix}
-47.0112  & -3.1933 & -13.2171
\end{bmatrix}. 
\end{equation}
In Fig \ref{fig:Example2}, we show the simulation of the obtained control law for the cases $N=2$, $N = 3$ and $N=4$. We can see that for this case, $N=4$ provides a close performance to the desired one.
\begin{figure}
\begin{center}
\includegraphics[width=8.4cm]{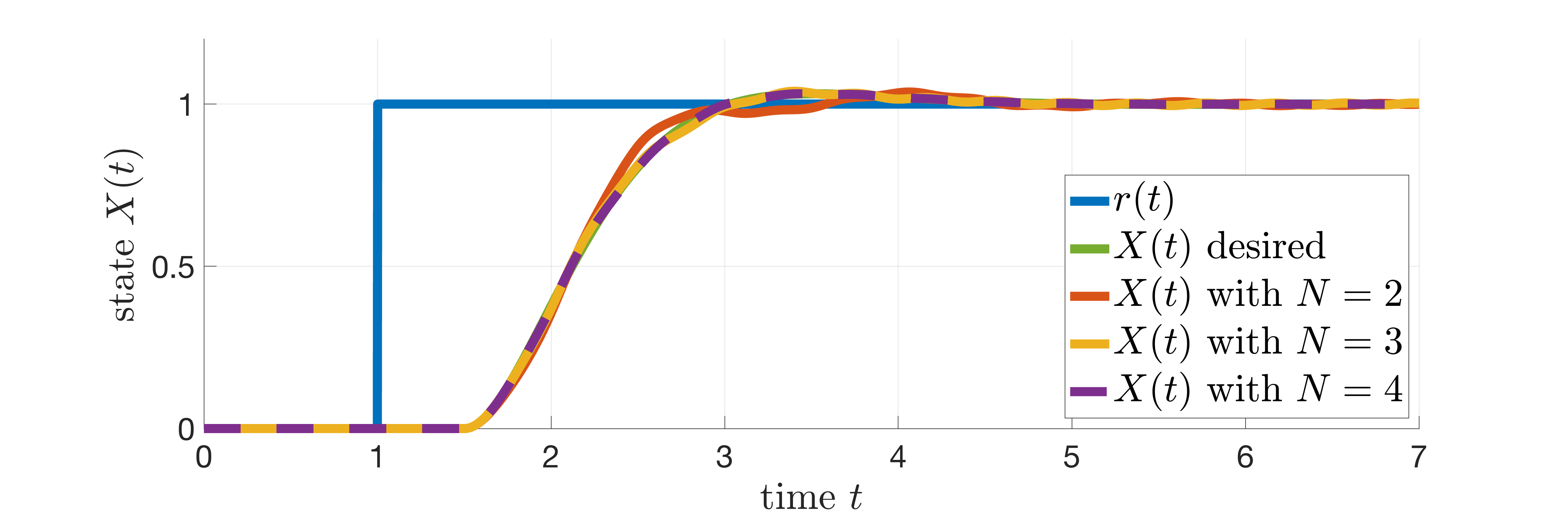}    
\caption{Closed-loop responses Example 2.} 
\label{fig:Example2}
\end{center}
\end{figure}

\subsection{Example 3. Unstable Continuous Stirred Tank Reactor \cite{Marquez2024JournalSimple}}\label{Example3}
We consider the unstable continuous stirred tank reactor proposed in \cite[Section~5.1]{Marquez2024JournalSimple}. The transfer function \cite[Equation~(27)]{Marquez2024JournalSimple} is written in a state space representation with the following matrices
\begin{equation*}
A = \begin{bmatrix} 
   -9.3310 &  -4.2220 &   2.1521\\
    4.0000 &        0 &        0\\
         0 &   4.0000 &        0
\end{bmatrix}, \quad 
B = \begin{bmatrix} 
0.0625 \\ 0 \\0\end{bmatrix}
\end{equation*}
\begin{equation*}
C = \begin{bmatrix} 
0  &      0  &  0.0646\end{bmatrix} 
\end{equation*}
and delay $\tau = D = 1.65$ ($\tau$ is the notation used in \cite{Marquez2024JournalSimple} for the delay). The matrix $K$ is designed using pole placement in such a way that the eigenvalues of $A+BK$ are placed in $\lambda = \begin{bmatrix} -0.5 + 1i &  -0.5 - 1i & -2\end{bmatrix}$. In Fig. \ref{fig:Example3}, we show the simulations for three different order of approximations. In this case, compared to \cite[Fig.~5]{Marquez2024JournalSimple}, we are able to reduce the oscillations and fix a desired behaviour in closed loop. This is possible through the flexibility and the extra degrees of freedom offered by the pole-placement when designing state feedback control instead of PID control.

\begin{figure}
\begin{center}
\includegraphics[width=8.4cm]{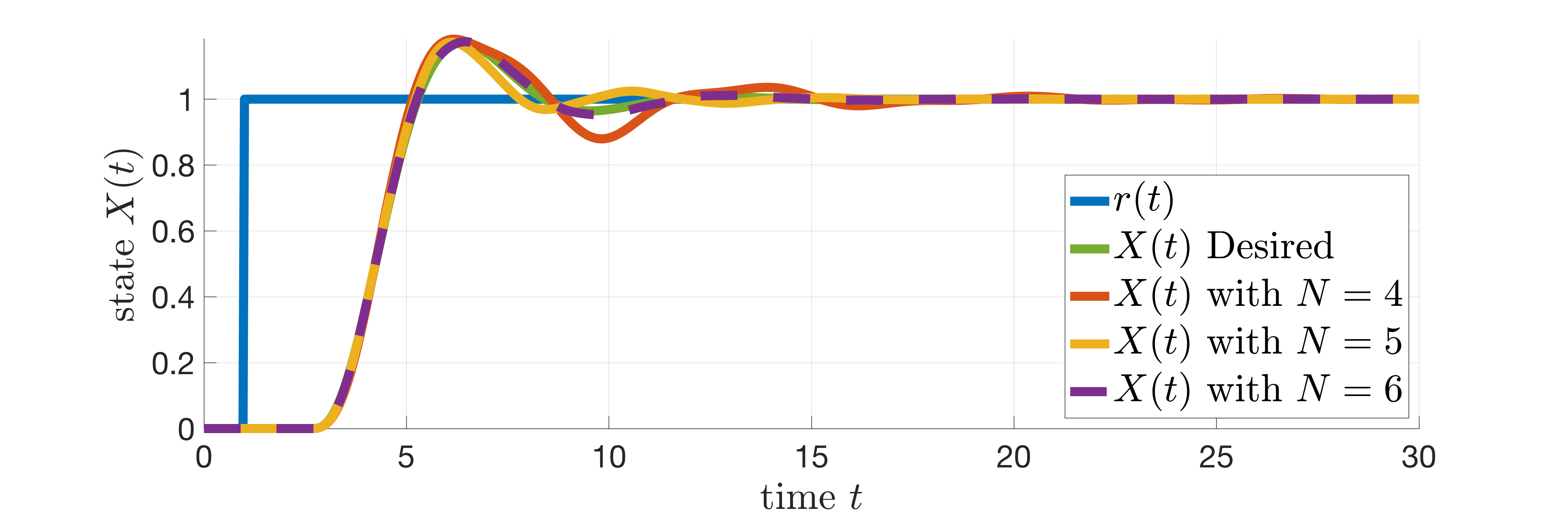}    
\caption{Closed-loop responses Example 3.} 
\label{fig:Example3}
\end{center}
\end{figure}

\begin{table}
\caption{Required $l$ for the stability test}\label{Table:l}
\begin{center}
\begin{tabular}{ |c|c| } 
 \hline
 Example 1 
  &  $l$ \\  \hline
 $N = 2$ & $4$ \\ 
 $N = 3$ & $4$  \\ 
 $N = 10$ & $7$  \\ 
 \hline
\end{tabular}, \begin{tabular}{ |c|c| } 
 \hline
 Example 2
  &  $l$ \\   \hline
 $N = 2$ & $5$ \\ 
 $N = 3$ & $6$  \\ 
 $N = 4$ & $5$  \\ 
 \hline
\end{tabular}, \begin{tabular}{ |c|c| } 
 \hline
 Example 3 
  &  $l$ \\   \hline
 $N = 4$ & $5$ \\ 
 $N = 5$ & $5$ \\ 
 $N = 6$ & $7$ \\ 
 \hline
\end{tabular}
\end{center}
\end{table}

\section{Conclusions and future work}\label{Sec:Conclusions}

In this paper, we present a new way to implement a backstepping controller for a linear system submitted to a known input delay. The proposed controller is given by a dynamic controller combined with a state feedback, which has been constructed as a dynamical approximation of the infinite dimensional control law. The closed-loop stability is guaranteed using a Lyapunov functional which has been shown to be non conservative at least asymptotically. Three examples taken from the literature have been used to show the proposed approach. A future work is the analysis of the robustness to uncertainties of the matrices $A$ and $B$ and on the delay $D$. \textcolor{black}{Another investigation line is the case in which the system is subject to measurements delays and a dynamical observer has to be used to re-build the state. Finally, the PDE-backstepping control law obtained for the wave equation interconnected in cascade with an ODE (See \cite[Section~16]{Krstic2009BookDelay}) is part of the current work of the authors.}

\appendix
\subsection{Orthogonality of Legendre Polynomials} \label{Sec:Appendix}
We use the orthogonality property of Legendre polynomials to show that
\begin{equation}\label{Eq:Orthogonality}
\int_0^D u(\zeta,t)^2 d\zeta \geq \sum_{k = 0}^{l-1} \tfrac{2k+1}{D} \Omega_k (t) ^2 = \tfrac{1}{D}\Omega(t)^\top Q \Omega(t)
\end{equation}
with $Q$ defined in \eqref{Eq:Q}.

\bibliographystyle{plain}
\bibliography{References2.bib}

\end{document}